\documentclass[journal]{IEEEtran}
\usepackage{float}
\usepackage{hyperref}
\usepackage{pgfplots}
\usetikzlibrary{arrows,shapes,positioning,shadows,trees}
\usepackage{graphicx}
\usepackage{algpseudocode}
\usepackage{algorithm}
\usepackage{algorithmicx}
\usepackage{multirow}
\usepackage{xcolor}
\usepackage{amssymb}
\usepackage{listings}
\usepackage{amsmath}
\usepackage{seqsplit}
\usepackage{multirow}
\pgfplotsset{compat=1.17}
\usepackage{amsthm}

\newtheorem{theorem}{Theorem}[section]

\newtheorem{corollary}{Corollary}[section]
\newtheorem{definition}{Definition}[section]
\newtheorem{lemma}[theorem]{Lemma}

\begin{document}

\title{Privacy-Preserving Password Cracking: \\ How a Third Party Can Crack Our Password Hash Without Learning the Hash Value or the Cleartext}

\author{Norbert Tihanyi, Tamas Bisztray, Bertalan Borsos, and Sebastien Raveau \thanks{Tamas Bisztray received funding from the Research Council of Norway (forskningsradet) under Grant Agreement No. 303585 (CyberHunt project), the EU Connecting Europe Facility (CEF) programme under Grant Agreement No. INEA/CEF/ICT/A2020/2373266 (JCOP project) and the Horizon Europe programme under Grant Agreement No. 101070586 (PHOENi2X project). The views and opinions expressed herein are those of the authors only and do not necessarily reflect those of the European Union or the European Research Council Executive Agency. Norbert Tihanyi was supported by the Ministry of Innovation and Technology of Hungary from the National Research, Development and Innovation Fund through the TKP2021-NVA Funding Scheme under Project TKP2021-NVA-29.}}



\maketitle
\begin{abstract}
Using the computational resources of an untrusted third party to crack a password hash can pose a high number of privacy and security risks. The act of revealing the hash digest could in itself negatively impact both the data subject who created the password, and the data controller who stores the hash digest. 
This paper solves this currently open problem by presenting a Privacy-Preserving Password Cracking protocol (3PC), that prevents the third party cracking server from learning any useful information about the hash digest, or the recovered cleartext. 
This is achieved by a tailored anonymity set of decoy hashes, based on the concept of predicate encryption, where we extend the definition of a predicate function, to evaluate the output of a one way hash function. The protocol allows the client to maintain plausible deniability where the real choice of hash digest cannot be proved, even by the client itself. The probabilistic information the server obtains during the cracking process can be calculated and minimized to a desired level.
While in theory cracking a larger set of hashes would decrease computational speed, the 3PC protocol provides constant-time lookup on an arbitrary list size, bounded by the input/output operation per second (IOPS) capabilities of the third party server, thereby allowing the protocol to scale efficiently. We demonstrate these claims both theoretically and in practice, with a real-life use case implemented on an FPGA architecture.

\end{abstract}


\section{Introduction}

\IEEEPARstart{P}ASSWORDS are the most widely used mechanism for knowledge based user authentication~\cite{ogorman_comparing_2003}. Although tech firms such as Apple, Google and Microsoft are pushing for a \textit{passwordless future} \cite{srinivas_one_2022}, the transition will not impact every domain of identity management as alternatives often fail to provide a set of benefits already present in passwords \cite{bonneau_quest_2012}. Consequently, passwords will remain a part of our every day life, especially in areas where it is not feasible to employ technologies required for passwordless authentication \cite{bonneau_quest_2012, siddique_biometrics_2017}. Therefore, related security and also \textit{password management} practices should be kept up to date \cite{grassi_digital_2017}.

Storing passwords in cleartext poses a high security risk as attackers upon compromising the system could learn not only the passwords, but the \textit{password choosing patterns} of individuals \cite{tihanyi_unrevealed_2015}. A possible mitigation is to only store the hash of the password \cite{kamal_security_2019}. Recovering the cleartext from a hash digest is very resource intensive and is often infeasible with sufficiently long and complex passwords. Unfortunately, \textit{password choosing habits} of individuals often lack such characteristics \cite{bonneau_passwords_2015}.
There are well-known software tools such as \textit{hashcat} or \textit{John the Ripper} that can be used to perform password cracking attacks efficiently. 
Advancements in modern \textit{GPU}s and \textit{FPGA}s have rendered previously popular hashing algorithms obsolete. The same is true for \textit{password policies} regarding recommended length, character set and complexity \cite{dellamico_password_2010, weir_testing_2010, egelman_passwords_2011, choong_united_2014}.

Testing password security can be conducted for \textit{legitimate purposes}.
The main motivation behind this paper originates from a penetration testing project that took place in $2020$. During a penetration testing engagement the \textit{Red Team} was able to retrieve an \textit{NTLM} hash of an important service account. It was known, that all service account passwords are randomly generated 9 character long strings containing uppercase and lowercase letters plus numbers. Because of the sensitive nature of the project, the \textit{Red Team} could not share the exact value of the \textit{NTLM} hash. According to the signed contract, the team was not allowed to reveal any cleartext passwords from the engagement to third parties. 
There are many small organizations and freelancers conducting \textit{penetration testing} and \textit{red teaming} activities, who would rely on such services, but are prevented from doing so for privacy considerations. 
If the \textit{Red Team} transfers a password hash, the cracking server learns the exact value of the hash digest and if the cracking is successful the corresponding \textit{cleartext password} as well.

As the example shows, an entity might have a legitimate reason to crack a password hash, but could lack the computational capacity to perform the cracking process within a reasonable time. 
One solution would be to use a cloud service called \textit{"password cracking as a service"} (\textit{PCaaS}), to utilize the resources of a third party.
This can be concerning both from a security, and a privacy perspective. As an example, hashes can be used in \textit{pass the hash} attacks \cite{ewaida_pass--hash_2010} without having to crack the password itself. Moreover, if the cleartext is recovered the third party might learn privacy sensitive information as individuals often embed \textit{personally identifying information} (\textit{PII}) when constructing their passwords \cite{li_personal_2017, wang_birthday_2019, veras_semantic_2014}.

We would like to find answers to the following questions: (1) \textit{can we use the resources of an untrusted third party for the cracking process without revealing the exact value of the target hash digest} and (2) \textit{if the hash is cracked, can we prevent information disclosure on the recovered cleartext}.
In theory, this could be prevented by employing \textit{homomorphic encryption} \cite{rivest_databanks_1987, alloghani_systematic_2019}, where the third party performs operations on encrypted data. The domain of \textit{homomorphic encryption} has been the subject to a lot of attention and there were many achievements that brought us closer to practical applications in the last decade \cite{parmar_survey_2014}. 
Unfortunately, as of today no such algorithm has the efficiency that could allow us to take advantage of this technology for the realization of our goals \cite{acar_survey_2018, alloghani_systematic_2019}. To the best of our knowledge prior to this publication privacy-preserving password cracking protocols have not been considered and documented in scientific literature. Such a protocol not only satisfies the client's security and privacy needs upon using \textit{PCaaS}, but can also assist in compliance with \textit{data protection and privacy regulation}, where the data controller can document justifiable and explainable privacy protection measures upon using this \textit{privacy enhancing technology} (\textit{PET}). 
The main contributions of this paper can be summarized as:

\begin{enumerate}
\item We introduce the idea of a \textit{Privacy-Preserving Password Cracking} (3PC) protocol 
\item We extend the concept of \textit{predicate functions} to evaluate decoy hash digests that make up our anonymity set, thereby resolving data transfer and performance issues
\item We show that the \textit{probabilistic information} the \textit{cracking server} learns is not practically useful, and the protocol is resistant against attacks and \textit{foul play}
\item The protocol provides \textit{plausible deniability}, where the client can claim to have aimed for a different target
\item Demonstrations of the implemented protocol, both with \textit{"toy examples"}, and realistic use cases, showing the scalability and efficiency of the protocol 
\item The protocol ensures the client, using \textit{proof of work}, that the server exhausted the agreed search space
\end{enumerate}
The paper will be structured as follows: Section \ref{literature} overviews related literature, while Section \ref{3pcprotocol} discusses the application of predicate functions, and introduces the 3PC protocol with two toy examples. Section \ref{probfoul} presents the security and privacy aspects of 3PC, followed by Section \ref{experimantal} where we showcase a real-life example, implemented on a FPGA architecture. Section \ref{summary} concludes our results.

\section{Related literature} \label{literature}
The idea to hide valid cryptographic keys and hashes among fake ones appeared in literature 20 years ago. Arcot \cite{hoover_software_1999} systems used a list of junk RSA private keys to protect the original private key. An attacker who tries to crack the key container will recover many plausible private keys, but will not be able to tell which one is the original until he tries each to access resources via an authentication server. In 2010 Bojinov et al. \cite{bojinov_kamouflage_2010} introduced the \textit{Kamouflage system}, a theft-resistant password manager which generates sets of \textit{decoy passwords}.

Juels and Ristenpart introduced the \textit{honey encryption} scheme \cite{juels_honey_2014}, designed to produce a ciphertext which, when decrypted with incorrect keys, produces plausible-looking but bogus plaintexts called \textit{honey messages}. This makes it impossible for an attacker to tell when decryption has been successful.
In 2013 Jules and Rivest proposed a simple method \cite{juels_honeywords_2013} for improving the security of hashed passwords. The system in addition to a real password with each user’s account associates some additional \textit{honeywords} (false passwords). The attempted use of a honeyword triggers an alarm. 

Compromised Credential Checking (C3) services, such as HaveIBeenPwned and Google Password Checkup can reveal if user credentials appear in known data breaches \cite{jennifer_pullman_protect_2019,thomas_protecting_2019}.
In this setup, clients can provide an N-bit prefix of the hashed password. The server then returns all recorded breached passwords that match this prefix. The client then conducts a local final check to confirm if there is a match.
However, in this scenario sharing a small hash prefix of user passwords can significantly increase the effectiveness of remote guessing attacks \cite{li_protocols_2019}. 
The anonymity set it provides will can only come from passwords of the same prefix, which might not be an adequately large set.
On the other hand if the prefix is too small, it might return a lot of results, which in this setting would be more difficulty to download for the user.
To counter this issue Li et al. developed and tested two new protocols that offer stronger password protection and are practical for deployment \cite{li_protocols_2019} .
This approach of only sharing a part of the password hash is useful, which we upgraded for the 3PC protocol to allow the creation of a more fine tuned anonymity set.

For our protocol, understanding what distribution passwords follow will be also important.
As Hou in \cite{hou_new_2022} underlined, several research papers incorrectly assume that passwords follow a \textit{uniform distribution}. In large real password data sets certain popular passwords are used by multiple users \cite{bonneau_science_2012}.
In \cite{blocki_differentially_2016} Blocki et al. presents how a \textit{frequency list} can be created based on this observation, essentially ranking passwords from most frequent to least frequent. Inspired by Malone \cite{malone_investigating_2011}, Wang et al. showed that such passwords follow a \textit{CDF-Zipf} distribution.
These results apply to large user generated data sets where recurrence of passwords can be observed.
If we want to analyze a large corpus of machine generated passwords results may be different. Such sets are usually created as a dictionary for password cracking. If the generation algorithm simply outputs random strings, the resulting data set by definition follows a uniform distribution.
However, password generation based on real user passwords is shown to perform better. The most popular techniques are: \textit{Rule-based dictionary} attacks, \textit{probabilistic context free grammars} (\textit{PCFG}s), \textit{Markov models}, and machine learning techniques \cite{aggarwal_new_2018}, \cite{weir_password_2009}. The likelihood of each password can be calculated if we assign a probability to each production rule. However, determining the probability distribution over such data sets is not a straightforward task and is outside the scope of this research. It is different from assessing word-list quality \cite{kanta_pcwq_2021,kanta_how_2021}, or individual password strength \cite{galbally_new_2017-1,wang_fuzzypsm_2016,oesch_that_2019,galbally_new_2017}, as both are a separate line of research. As Aggarwal et al. shows \cite{aggarwal_new_2018}, \textit{PCFG} \textit{parse trees} are usually \textit{ambiguous} which means, there can be multiple ways to produce a single word. 
Furthermore, if a capable adversary aims to determine the probability distribution of such data sets but doesn't know the original creation mechanism, they can get a completely different ranking order if for example a \textit{Markov model} is used to rank a data set created by \textit{PCFG}. Not to mention that each model is fully dependent on the original training data.
Note, that Bonneua warns against the use of traditional metrics such as \textit{Shannon entropy} and \textit{guessing entropy} for evaluating large password data sets \cite{bonneau_science_2012}.

For hiding the original hash digest among decoy hashes we will utilize a concept similar to predicate encryption. In predicate encryption \cite{gennaro_predicate_2015}, a ciphertext is associated with descriptive attribute values $\upsilon$ in addition to a plaintext $p$, and a secret key is associated with a predicate $f$. Decryption returns plaintext $p$ if and only if $f(v) = 1$. In our case there will be no secret key associated with the predicate, as hash functions cannot be reversed in the same way as an encryption function but as the concept is similar we use this terminology.

\section{The 3PC protocol} \label{3pcprotocol}

First, we review some important cryptographic primitives and their properties we rely on, which is followed by the main steps of the 3PC protocol. Next, we introduce how a predicate function can evaluate the output of a hash function. Finally, we examine how the protocol works in action by showcasing two \textit{"toy examples"} serving as proof of concept.

A hash function is a computationally efficient deterministic function mapping from an arbitrary size input (\textit{message space}) into a fixed size output of length $l$ (\textit{digest space}). 
A hash function is called a \textit{cryptographic hash function} if the usual security properties are satisfied: \textit{pre-image resistance} (or one-way property), \textit{second pre-image resistance}, and \textit{collision resistance}. We also require that a hash function exhibits the \textit{avalanche effect} and satisfies the \textit{random oracle model} where the output is uniformly distributed.
It is important to note that theoretically there are infinitely many collisions for a hash function, but these should be difficult to find, meaning, there is no explicit or efficient algorithm that can output a collision.
 \subsection{Protocol Design and Notations}
For the rest of the paper the following notation will be used:

\begin{itemize}
\item{$t$: The \textit{target hash} which the client wants to crack}

\item{$\Sigma$: the finite set of hexadecimal symbols \{0-F\}, where the finite sequence of symbols of length $l$ over the alphabet $\Sigma$ is denoted by $\Sigma^l$}

\item{$\Theta^{*}$: For any alphabet $\Theta$, the set of all strings over $\Theta$}
\item{$h$: A cryptographically secure \textit{hash function} where $\Theta^{*}$  serves as the input space for $h$,  i.e.,  $h:\Theta^*\to\Sigma^l$}
\item{$\mathcal{X_\upsilon}$: Set of \textit{decoy hashes}, serving as an anonymity set for the target hash, where $t\in\mathcal{X_\upsilon}$ }
\item{$N_\upsilon$: The desired number of \textit{decoy hashes} the client wants as the size of $\mathcal{X_\upsilon}$ }
\item{$\upsilon$: A vector that describes every hash in $\mathcal{X}_\upsilon$, and $\upsilon \in \Sigma^{2l}$}
\item{$\mathcal{DS}$: The \textit{cracking data set}, selected by the client and used by the  cracking server to recover hashes from $\mathcal{X}_\upsilon$}
\item{$\mathcal{CS}$: \textit{Candidate password set}, a set of password-hash pairs the server managed to crack from $\mathcal{X}_\upsilon$}
\end{itemize}
A high level pseudo code of the 3PC protocol can be seen in Algorithm \ref{3pcalgo}. 
Our two parties are the client and the server, as indicated below in Figure \ref{example}. The client wants to crack the \textit{target hash}, and recover the cleartext password. First, the client checks using REQ-H, if the server side architecture can efficiently crack the desired hash function. The server responds with ACK-H, specifying the hash cracking speed in hash-rate/secundum (H/Sec). 
Based on this information, using DEF-DSR the client selects a \textit{data set} $\mathcal{DS}$, and specifies the desired number of \textit{candidate passwords} which is denoted by $r$. Meaning, the \textit{candidate set} $\mathcal{CS}$ should contain approximately $r$ elements. 
Next, the client creates a vector $\upsilon$ that defines the set of decoy hashes, such that $|\mathcal{X}_\upsilon| \approx N_\upsilon$ and $t \in \mathcal{X}_\upsilon$ is satisfied (CLC-NV, GEN-V). 
We note, that creating a vector $\upsilon$ which defines exactly $N_\upsilon$ elements is linked to the theory of $y$-smooth numbers. As a consequence $N_\upsilon$ can only be approximated when $\upsilon$ is created.
The \textit{cracking data set} which can be a selection of dictionary words with mangling rules, brute force rules, etc., and the vector $\upsilon$ is transferred to the server (SND-P). 
The server starts hashing every password in $\mathcal{DS}$ (CRK-XV). If a resulting hash digest is in the set defined by the vector $\upsilon$ (this is checked by the predicate function, not by a direct comparison), this hash digest along with the corresponding \textit{cleartext password} is added to the list of \textit{candidate passwords}. After exhausting the search space the server sends $\mathcal{CS}$ back to the client (SND-CS).
\begin{figure*}[ht]
\centering
\includegraphics[width=0.7\textwidth]{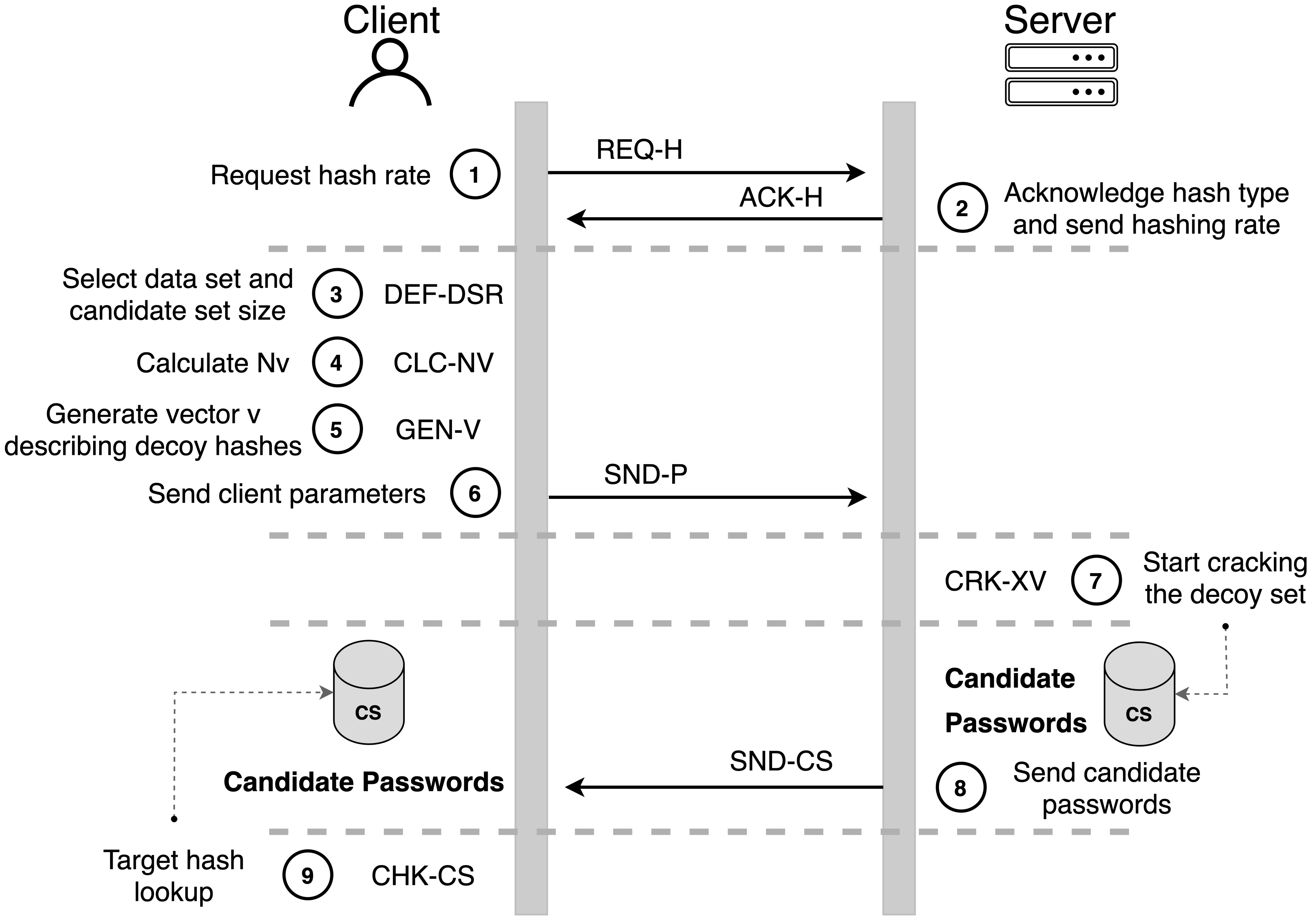}
\caption{Client-Server interaction. A simplified version on how information exchange is made between the client and the third party server}
\label{example}
\end{figure*}
To evaluate if the cracking was successful, the client simply checks the \textit{candidate list} to see if $(s,t) \in \mathcal{CS}$ for some $s \in \mathcal{DS}$ (CHK-CS). 

\noindent Essentially, the server tries to crack every password in the \textit{decoy set}, never knowing which is the \textit{target hash}. We want to underline, that the server at no point has any information on whether $t$ is cracked or remains uncracked, which is of utmost importance. This is due to the protocol design, where it is close to impossible to crack all hashes in the \textit{decoy set}.

\begin{algorithm}
  \caption{3PC protocol Client-Server exchange}\label{p3pcrack}
  \begin{algorithmic}[1]
    \Procedure{3PC}{$t,r, \mathcal{DS}$}
    \State REQ-H($h$)  \Comment{Request hash information}
    \State $i \gets$ ACK-H($h$)  \Comment{Acknowledge hash info request   }
    \State $\mathcal{DS},r \gets$ DEF-DSR($i$)  \Comment{Define data set and $r$}    
    
    \State $N_\upsilon \gets$ CLC-NV($r,|\mathcal{DS}|$) \Comment{Calculate $N_\upsilon$}
     \State $\upsilon \gets$ GEN-V($t,N_\upsilon$) \Comment{Generate decoy hashes}
     \State SND-P($\mathcal{DS},\upsilon$) \Comment{Send  parameters to server}

     \State $\mathcal{CS} \gets$ CRK-XV($\upsilon,\mathcal{DS}$) \Comment{Cracking decoy set}
      \State SND-CS($\mathcal{CS}$) \Comment{Send candidate passwords}
      \State CHK-CS($\mathcal{CS},t$) \Comment{Check target hash in candidate set}
    
\EndProcedure
  \end{algorithmic}
   \label{3pcalgo}
\end{algorithm}

The decoy hash set will serve as an anonymity set, providing 
\textit{k-anonymity} for the \textit{target hash digest}, and the server can only have a $1/|\mathcal{X}_\upsilon|$ chance to guess $t$ when observing $\mathcal{X}_\upsilon$ alone, which in practice can easily be around $10^{-70}$. 
The decoy hashes at the same time as protecting $t$, can ensure the privacy protection of the \textit{pre-image} of $t$, regardless of whether it is cracked or not. This is a really important claim. As the server recovers many plausible cleartext passwords these will serve as an anonymity set for the \textit{pre-image} of $t$. 
Using the \textit{security parameters} $\mathcal{DS}, N_\upsilon$ and $r$, one can calculate and adjust the probabilistic knowledge of the server about both the \textit{target hash}, and the \textit{cleartext password} that created the hash. This knowledge can be tailored based on the calculated privacy needs of the client, or the resources of the third party.

Privacy demands could require ever larger sets of decoy hashes, where transferring $|\mathcal{X}_\upsilon|$ unique hashes would be a serious bottleneck.
In addition to the privacy challenges, this problem is also solved by the 3PC protocol, as the decoy set is transferred in a compact form described by the vector $\upsilon$. This, at the same time allows the \textit{predicate function} to remove limitations on the amount of hashes we can efficiently handle. The client or the server side never needs to write the entire $\mathcal{X}_\upsilon$ set to disk, which could be a bottleneck, as in realistic scenarios it can contain $10^{70}$ unique hash digests or more. This in itself is an important contribution, allowing the CRK-XV process to finish regardless of the number of decoy hashes in $\mathcal{X}_\upsilon$ as we will show in Lemma \ref{ordo1}. 
\subsection{Establishing the Security Parameters}

In the following, we examine how the client can predictably control the number of hashes that are cracked from the \textit{decoy set} $\mathcal{X}_\upsilon$. Since the \textit{pre-images} of the decoy hashes are fully random, this "cracking success rate" will be determined only by the size of the \textit{cracking data set} $|\mathcal{DS}|$ and $|\mathcal{X}_\upsilon|$. What passwords are present in the \textit{data set} will have no influence on this calculation. To understand this, we need to look at how password hashes are distributed over the \textit{co-domain} of $h$.
When elements of $\mathcal{DS}$ are hashed, the output is spread randomly over $\Sigma^l$ as shown in Figure \ref{hash2}. 
This means that by picking a \textit{hash digest} randomly from the output space of $h$, the probability of selecting one that has a \textit{pre-image} from $\mathcal{DS}$ is simply $|\mathcal{DS}|/|\Sigma^l|$. As there is no correlation between the \textit{cleartexts} and the \textit{hashes}, by observing only the output space, selecting any of the hash digests carries a $|\mathcal{DS}|/|\Sigma^l|$ chance to have a \textit{pre-image} in $\mathcal{DS}$. By picking $k$ hash digests, the expected number of them having a \textit{pre-image} in $\mathcal{DS}$ is $k$ times $|\mathcal{DS}|/|\Sigma^l|$. 

This expected number is the same whether $k$ hashes are picked from $\mathcal{X}_\upsilon$, or by selecting $k$ randomly from $\Sigma^l$.
Thus we can choose from $\mathcal{X}_\upsilon$, the set described by $\upsilon$, where all hashes by definition satisfy the predicate function, which is introduced in Section \ref{ysmooth}.
Now, we can calculate the expected number of \textit{candidate passwords} ($r$) we get from a given $\mathcal{X}_\upsilon$ when cracking with a \textit{data set} with size $\mathcal{|DS|}$:
\begin{equation}
N_\upsilon \frac{|\mathcal{DS}|}{|\Sigma^l|} \approx r  
\label{equation2}
\end{equation}
This formula is used by the client in CLC-NV and GEN-V, to determine $N_\upsilon$, and then to check if vector $\upsilon$ can ensure approximately $r$ candidate passwords in $\mathcal{CS}$. 
If we want to have more \textit{decoy hashes} in $\mathcal{X}_\upsilon$ which upon getting cracked will reveal a \textit{pre-image} in $\mathcal{DS}$, we simply need to increase $|\mathcal{X}_\upsilon|$. 
\begin{figure}[ht]
\centering
\includegraphics[width=0.4\textwidth]{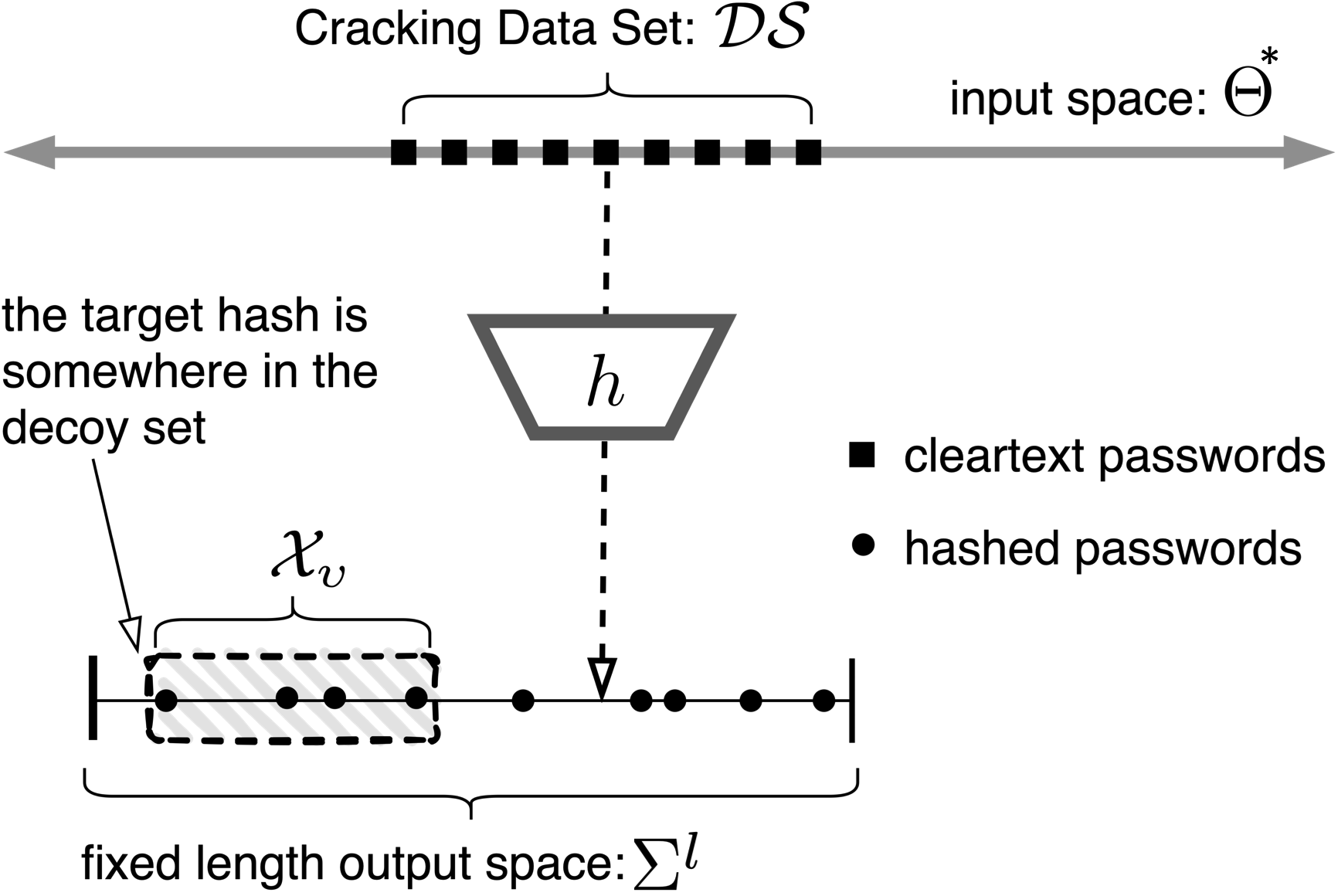}
\caption{Hashing the cracking data set and selecting an anonymity set $\mathcal{X}_\upsilon$ for a target hash $t$}
\label{hash2}
\end{figure}

Visually, we can see this from in Figure \ref{hash2}, that as we expand the $\mathcal{X}_\upsilon$ set, it will "gobble up" more of the dots representing the hashed $\mathcal{DS}$ words. 
The \textit{target hash} is only cracked if the \textit{cracking data set} used contains a \textit{pre-image} for $t$. Increasing $\mathcal{X}_\upsilon$ will achieve two things: it produces more candidate passwords for a given \textit{data set} selected, and it decreases the probability of a correct guess on the \textit{target hash digest}.
Due to the random distribution of the hash function the security parameter $r$ is just an expected value. Although we can reliably estimate it, we will not know the exact number of candidates that will be returned in the \textit{candidate set} until the end of the cracking process. Also, the \textit{server} will never know if a certain \textit{data set} has found the \textit{target hash}, or it is still uncracked.

\subsection{Mathematical foundations of 3PC}\label{ysmooth}
An important contribution of this paper which we next introduce, is the extension of the concept of predicate encryption to one way hash functions. Following that, we present the underlying mathematical problem that needs to be solved during the step GEN-V.

\begin{definition}[\textbf{Predicate Function}]

The set of all strings of length 
$l$ over the alphabet $\Sigma$ is denoted by $\Sigma ^{l}$, where our alphabet will be the set of hexadecimal characters. 
For a vector $\upsilon=(\upsilon_1,\dots,\upsilon_{2l}) \in \Sigma^{2l}$ we define a predicate function $P_\upsilon$ over $\Sigma ^{l}$ as follows: Given an input vector $x=(x_1,\dots,x_l) \in \Sigma^l$, $P_\upsilon(x)=1$ if $(\upsilon_{2i-1} \leq x_i \leq \upsilon_{2i})$  for all  $1 \leq i \leq l $, $P_\upsilon(x)=0$ otherwise.
\end{definition}

\noindent For example if $x=(3,C) \in \Sigma^{2}$ and $v=(2,5,C,D) \in \Sigma^{4}$, we have $P_\upsilon(x)=1$ since $2\leq x_1\leq 5$ and $C\leq x_2\leq D$ are satisfied for our vector $\upsilon$. 
For a given $\upsilon \in \Sigma^{2l}$, let $\mathcal{X_\upsilon}$ be the set of all $x$ vectors which satisfy $P_\upsilon(x)=1$. The number of elements (cardinality) of $\mathcal{X_\upsilon}$ is denoted by  $|\mathcal{X_\upsilon}|$. The size of $\mathcal{X_\upsilon}$ can be calculated from  the definition of $\upsilon$ by the following formula:
\begin{equation} 
\label{sizeofXv}
 |\mathcal{X_\upsilon}|=\prod\limits_{i=1}^{l} max\{\upsilon_{2i}-\upsilon_{2i-1}+1,0\}
\end{equation}

\noindent Clearly, we have a set of  $|\mathcal{X_\upsilon}|$ different vectors making up $\mathcal{X_\upsilon}$, i.e.,

\[ \mathcal{X_\upsilon} = \left[ {\begin{array}{ccccccccccccc}
x^1=(x^1_1,\dots,x^1_l) \\
x^2=(x^2_1,\dots,x^2_l) \\
 \vdots \\ 
x^{|\mathcal{X_\upsilon}|}=(x^{|\mathcal{X_\upsilon}|}_1,\dots,x^{|\mathcal{X_\upsilon}|}_l) \\
\end{array} } \right] \]

\noindent Note, that $|\mathcal{X_\upsilon}|=1$ iff $\{ \upsilon_{2i}=\upsilon_{2i-1}: \forall \,1 \leq i \leq l \}$ and  $|\mathcal{X_\upsilon}|=0$ iff $\exists \,i$ such that  $\upsilon_{2i-1}>\upsilon_{2i}$.
Similarly, if for all $\,1 \leq i \leq l $, $\upsilon_{2i}$ is the greatest element of $\Sigma$ and $\upsilon_{2i-1}$ is the least element of $\Sigma$, then $|\mathcal{X_\upsilon}|=|\Sigma^l|$. 
\begin{lemma} 
The calculation of $P_\upsilon(x)$ is independent of the cardinality of $\mathcal{X}_\upsilon$.
\label{ordo1}
\end{lemma}

\begin{proof}
From condition $\{\upsilon_{2i-1} \leq x_i \leq \upsilon_{2i}: 1 \leq i \leq l\}$, $P_\upsilon(x)$ can be calculated in constant time which in the worst case is $2\cdot l$ comparison.

\end{proof}
We say an integer $N$ is $y$-smooth if $N$ has no prime divisors greater than $y$.
\begin{lemma} \label{smoothlemma}
$|\mathcal{X_\upsilon}|$ is always a $13$-smooth number.
\end{lemma}
\begin{proof}

For any $\upsilon=(\upsilon_1,\dots,\upsilon_{2l}) \in \Sigma^{2l}$ vector it is trivial that $0 \leq \upsilon_i \leq |\Sigma|$ for all $1 \leq i \leq l$. From equation (\ref{sizeofXv}), the prime factors of $|\mathcal{X_\upsilon}|$ are always less than $|\Sigma|=16$, from which the lemma follows.
\end{proof}
\begin{corollary} \label{lemma2}
The prime factorization of every $|\mathcal{X_\upsilon}|$ number is of the form 
$2^A\times 3^B \times 5^C \times 7^D \times 11^E \times 13^F$, where ${A,\dots, F}\in \mathbb{N}.$
\end{corollary}

During step CLC-NV, an appropriate $N_\upsilon$ is calculated from the size of the \textit{data set} $|\mathcal{DS}|$, and the desired number of \textit{candidate passwords} $(r)$ using equation \ref{equation2},

\begin{equation}
    N_\upsilon = \frac{r |\Sigma^l|}{|\mathcal{DS}|} 
    \label{NVformula}
\end{equation}
In the ideal scenario, the client can create a vector $\upsilon$ that defines the set of decoy hashes $\mathcal{X}_\upsilon$ with exactly $ N_\upsilon$ elements, i.e., $|\mathcal{X}_\upsilon| =||N_\upsilon||$ where $||\cdot||$ denotes the nearest integer. However, $N_\upsilon$ is not necessarily $13$-smooth, so we need to find a $13$-smooth number close to the original $N_\upsilon$. We can take the logarithm of both sides of the equation $|\mathcal{X}_\upsilon| =||N_\upsilon||$, i.e.,  
\begin{equation*}
    \log(2^A\times 3^B \times 5^C \times 7^D \times 11^E \times 13^F)=\log(N_\upsilon)
\end{equation*}
from which we get
\begin{equation}
\begin{array}{cc}
       A\log(2)+B\log(3)+C\log(5)+ &\\
      D\log(7)+E\log(11)+F\log(13) & \\
      - \log(N_\upsilon)=0
\end{array}    
\end{equation}
This is a $7$-variable \textit{integer relation problem}, and can be viewed as a special subset sum optimization problem. Using the Lenstra, Lenstra and Lov\'asz (LLL)  basis reduction algorithm \cite{lenstra_factoring_1982}, or the $PSLQ$ algorithm, one can find the fitting integers. We note, that not all $13$-smooth numbers are suitable for our needs, therefore we need to run $LLL$ for different inputs with varying beta reduction parameter. For example, although $N_\upsilon=5^{l+1}$ is $13$-smooth, it is not possible to divide the factors into $l$ slots, where $\{ \upsilon_{2i}-\upsilon_{2i-1}\leq16: \forall \,1 \leq i \leq l \}$ stays true. Modern $LLL$ implementations can solve integer relation problems with more than 500 variables. In our case we can always get an appropriate result in polynomial time (in a few seconds on an average computer).

As previously stated, our objective is to crack the \textit{target hash} and hide it in a \textit{decoy set} represented by $\upsilon$. As $P_\upsilon(t)=1$ must be satisfied we construct $\upsilon$ from the \textit{target hash}. Depending on how we adjust the degree of freedom in vector $\upsilon$, it can allow $|\mathcal{X}_\upsilon|-1$ decoy hashes to satisfy the function $P_\upsilon$. These vectors will serve as decoy hashes.

Now we possess the minimum knowledge to understand the protocol. First, we present 3PC through two toy examples, before moving on with the security analysis.

\subsection{Toy Examples}

\textbf{Toy Example 1 - Dictionary attack:} The \textit{Rock-You} database contains $14\,344\,391$ unique passwords which will be  our $\mathcal{DS}$ for this example \cite{tihanyi_unrevealed_2015}. The target hash is $t=(C,6,B,F,A,B,A,2)\in\Sigma^8$, where $t$ is the \textit{CRC-32} output of our unknown password.
CRC-32 is not a cryptographically secure hash function however, it's presentable output size is more suitable for a toy example. After requesting hashing information (REQ-H) and getting acknowledgement (ACK-H) from the server, the client is looking to create a vector $\upsilon$ that defines an $\mathcal{X_\upsilon}$ set, such that after the cracking process approximately 20 candidate passwords are returned in $\mathcal{CS}$. 
Knowing the size of the dictionary $|\mathcal{DS}|=14\,344\,391$, and the security parameter $r=20$, we can calculate $N_\upsilon$.
This means that after $14\,344\,391$ hash calculations the \textit{client} expects approximately 20 candidates from the \textit{Rock-You} database to fall into $\mathcal{X_\upsilon}$. To satisfy this requirement, step CLC-NV uses formula \ref{equation2} to calculate $N_\upsilon$ (the expected size of $\mathcal{X}_\upsilon$),
\begin{equation}
    N_\upsilon = \frac{r |\Sigma^l|}{|\mathcal{DS}|}  = \frac{}{} \frac{20 \cdot 16^8}{14\,344\,391}
\end{equation}

The result is $N_\upsilon \approx 5988.36$. 
Having calculated $N_\upsilon$, we need to construct a vector $\upsilon$ where $P_\upsilon(t)=1$ and $|\mathcal{X}_\upsilon|$ is around this size. Note, that $5988$ is not 13-smooth as $5988=2^2 \cdot 3 \cdot 499$. Using the GEN-V algorithm a suitable $\upsilon$ can be selected: $(C,F,2,6,A,B,D,F,9,F,B,B,A,A,0,6) \in \Sigma^{16}$. This defines an $\mathcal{X_\upsilon}$ decoy set with $5880$ elements, which satisfies Lemma \ref{smoothlemma}, Corollary \ref{lemma2}, and is close to $N_\upsilon$. Moreover, the degrees of freedom can be divided into $l$ slots, while satisfying $\upsilon_{2i}-\upsilon_{2i-1}\leq16$, $\forall \,1 \leq i \leq 8$. As such, it will produce nearly the same number of expected candidates:

\begin{equation}
|\mathcal{X}_\upsilon| \frac{|\mathcal{DS}|}{|\Sigma^l|} =  5880\frac{14\,344\,391}{ 16^8} \approx 19.63
\end{equation}

After receiving SND-P the server starts hashing every element in the dictionary, and it looks for matches in the set $\mathcal{X_\upsilon}$. Here lies one of the major benefits of the protocol: the server never has to make $|\mathcal{X}_\upsilon|$ comparisons for each password. After calculating the hash, the server simply checks if $P_\upsilon(h(s))=1$, $\forall s \in \mathcal{DS}$. In this toy example we expect around 20 passwords from $\mathcal{DS}$ to satisfy $P_\upsilon$. 
The \textit{hash-cleartext} pairs where $P_\upsilon(h(s))=1$ can be seen in Table \ref{toyexample1}.   
\begin{table}[H]
\caption{Toy Example 1: Generated candidate hashes }
\centering
\begin{tabular}{|c|c|c|c|c|c|} 
\hline
\textbf{\#} & \textbf{password} & \textbf{CRC-32} & \textbf{\#} & \textbf{password}                     & \textbf{CRC-32}                     \\ 
\hline
1           & tangan        & \texttt{C5AEFBA5}        & 11           & 0849831211       & \texttt{D4BFDBA6}                            \\ 
\hline
2           & hornbyneho        & \texttt{C3AEFBA0}        & 12           & lumpibuniz       & \texttt{E3ADEBA4}                            \\ 
\hline
3           & 28707adnen    & \texttt{C4AE9BA4}        & 13           & sweep21     & \texttt{E3BEEBA6}                            \\ 
\hline
4           & lisa1842         & \texttt{C4BECBA2}        &14            & 577672   & \texttt{E3BEFBA5}                            \\ 
\hline
5           & \textbf{0BChrist} & \textbf{\texttt{C6BFABA2}}& 15          & 050462654 &\texttt{E5BDBBA1}                             \\ 
\hline
6          & sapphire24          & \texttt{C6BFDBA2}        & 16      & horses33 & \texttt{F2BEBBA0} \\ 

\hline
7          & Kissarmy1!           & \texttt{D2ADFBA6}        & 17          & zuzuloka   & \texttt{F3AECBA1}                            \\ 
\hline
8          & whateva89          & \texttt{D2BE9BA3}        & 18          & ms.jackson2008      & \texttt{F3BEDBA5}                            \\ 
\hline
9          & keno333$\_$          & \texttt{D3BDABA3}        & 19          & a2gfamilymaster  & \texttt{F5BDDBA5}                           \\ 
\hline
10          & bighottie          & \texttt{D4AEABA2}        & 20          & alana123456789    & \texttt{F6ADABA2}                            \\            
\hline
\end{tabular}
\label{toyexample1}

\end{table}

From the returned candidate set the client can see that the \textit{target hash} is cracked and the password is \texttt{"0BChrist"}. 
The server was expecting around 20 candidates from the dictionary, but has no idea if the cracking was successful.   
It is also important to ask the question if the passwords are equally likely in this case.
This is thoroughly examined in Section\ref{probfoul}. For the sake of this example we selected a password from \textit{Rock-You} to begin with, hence it was present in the candidate set.

\textbf{Toy Example 2 - Brute force:}
A more realistic case is to use a cryptographically secure hash function. In this scenario the \textit{client} has the extra knowledge, that the hash digest hides an \textit{8 digit PIN code}. The client must consider, that leaking such information could significantly improve the guessing ability of the third party as we will discuss in Section \ref{probfoul}. 

Let  $t=(B,2,3,B, \cdots,  A,7,9,3) \in \Sigma^{64}$ represent the \textit{SHA}-\textit{256} output of the following PIN code: \texttt{"43256891"}, where the \textit{cleartext} is not known by the \textit{client}. There are $10^8$ different 8 digit number codes which will be our $\mathcal{DS}$.
The client determines $r=10$ to be the expected number of \textit{candidate passwords} (which would be once again, an insufficient amount in a realistic scenario).
To find a vector $\upsilon$ that satisfies this the client first performs step CLC-NV: 
\begin{equation*}
    N_\upsilon = \frac{r |\Sigma^l|}{|\mathcal{DS}|}  = \frac{}{} \frac{10 \cdot 16^{64}}{10^8}
\end{equation*}

The number of \textit{decoy hashes} should be $N_\upsilon \approx 1.158*10^{70}$. By using the GEN-V algorithm we can generate a suitable $\upsilon=(\upsilon_1,\dots,\upsilon_{128}) \in \Sigma^{128}$ vector, where $P_\upsilon(t)=1$ is satisfied and $|\mathcal{X}_\upsilon|$ is near $N_\upsilon$. A suitable vector can be $\upsilon=$\texttt{\seqsplit{[7c27385c3f3f3f3f3f0c3f3f3f3f3f0c3f0c3f3f0c3f3f3f3f3f0c0c3f0c0c3f3f3f3f0c3f0c0c0c0c3f0c0c3f3f0c0c0c3f0c3f0c0c3f0c0c3f0c0c3f0c0c3f]}} $ \in \Sigma^{128}$.

After the server calculates the SHA-256 hashes for all the $10^8$ different number codes, 9 different PINs are cracked from the set $\mathcal{X_\upsilon}$. By observing the candidates the \textit{cracking server} still cannot know which is the password or if it is cracked at all. The cracked \textit{candidate passwords} can be seen in Table \ref{sha256}:

\begin{table}[h] 
\caption{Toy Example 2: Generated candidate hashes }
\centering
\begin{tabular}{|c|c|c|c|c|c|} 
\hline
\textbf{\#} & \textbf{password} & \textbf{SHA-256}  \\ 
\hline
1           & 15851680            & \texttt{85869d73ebe4c562cbde168898669053\dots}                  \\ 
\hline
2           & 18662804            & \texttt{a58bbf75d9cad8fc764cb3f364823a3b\dots}          \\ 
\hline
3           & 28251765            & \texttt{b26c78a4916d348565d986d4a6926034\dots}    \\ 
\hline
4           & 36823110            & \texttt{b27ccbfa99fc96dc38a445accd40d148\dots}    \\ 
\hline
5           & 37012370            & \texttt{945ab8ad984bfcb38b4f36cc36b73cae\dots}    \\ 
\hline
6           & \textbf{43256891}   & \texttt{b23be566408ad8d2f1ac0d84330c3127\dots}    \\ 
\hline
7           & 56995169            &  \texttt{7366edc5bc43387536ba6f47ad2ac834\dots}  \\
\hline
8           & 60409880            & \texttt{b689a9c7a5d539c8abfc197ae87a705e\dots}    \\ 
\hline
9           & 98509815            & \texttt{b3859a5f5ccfef995bd723c35598d137\dots}    \\ 
\hline
\end{tabular}
\label{sha256}
\end{table}

Selecting such a small candidate set size and exhausting the search space for a specific input can raise several concerns from a privacy point of view which we will discuss in connection with this example in Section \ref{probfoul}. 
One could argue that for this toy example it would not be necessary to use a \textit{PCaaS} third party, as \textit{SHA}-\textit{256} is a hash function is fast to compute, therefore, anyone could hash $|\mathcal{DS}|=10^8$ passwords on their laptop using \textit{John the Ripper}. This is however not the case if we change the hash function to \textit{bcrypt}. If combined with proper key stretching techniques (e.g: $2^{16}$ iterations) \textit{bcrypt} would make it infeasible to brute force even such small key-spaces on desktop computers.

\section{Security and Privacy analysis of 3PC} \label{probfoul}
This section considers the knowledge the third party cracking server possesses, or information an attacker could gain upon acquiring the vector $\upsilon$, considering different side channel attacks. Finally, we examine how a malicious attacker could try to tamper with the results and how such situations are evaded by the protocol design.

\subsection{Probabilistic knowledge without assumptions} \label{prob}
At first, we examine the best guessing strategy for the server without any assumption on the probability distribution of passwords, as they can vary from uniform (in case of machine generated) to variations of Zipf distribution (for human generated).
Since by definition $t \in \mathcal{X}_\upsilon$, even without starting the cracking process a correct guess on the \textit{target hash} can be made with $1/|\mathcal{X}_\upsilon|$ probability. \textit{What can the server at this point know about the potential cleartext?}
By design, the cracking process can only be successful if the selected \textit{data set} contains a cleartext $s^*$ such that $h(s^*)=t$.
Let $\mathcal{A}$ be the event the server guesses $s^*$ correctly, and $\mathcal{B}$ that $s^* \in \mathcal{DS}$. We can calculate the conditional probability of guessing $s^*$ right:

\begin{equation}
\mathcal{P}(\mathcal{A})=\mathcal{P}(\mathcal{A}|\mathcal{B}) \cdot \mathcal{P}(\mathcal{B})
\end{equation}
Note, that $s^* \in \mathcal{DS}$ is not an assumption but a condition. If this condition is satisfied the server is faced with the following situation. Without spending any resources, the best guess would be $1/|\mathcal{DS}|$. Note, that at this stage there is no point in examining the probability distribution over $\mathcal{DS}$. The predicate function selects every password in the \textit{data set} uniformly with exactly $|\mathcal{DS}|/|\mathcal{X}_\upsilon|$ probability, and only a random subset of $s\in\mathcal{DS}$ satisfies $P_\upsilon(h(s))=1$, in other words:

\begin{equation}
\mathcal{P} \Bigl( P_\upsilon(h(s))=1 \Bigr) =\frac{|\mathcal{DS}|}{|\mathcal{X}_\upsilon|},  \forall s\in \mathcal{DS}.
\end{equation}
The server can greatly increase it's guessing ability by sorting out all passwords where $P_\upsilon(h(s))=0$ as shown in Figure \ref{sets}.

\begin{lemma} \label{guesslemma}
If condition $\mathcal{B}$ is fulfilled and $|\mathcal{X}_\upsilon| < |\Sigma^l|$, the server must hash $\forall s\in \mathcal{DS}$ to maximize $\mathcal{P}(\mathcal{A}|\mathcal{B})$ over the given \textit{data set}.
\end{lemma}
\begin{proof}
If $|\mathcal{X}_\upsilon| < |\Sigma^l|$, then only a random subset of $s\in \mathcal{DS}$ will satisfy $P_\upsilon(h(s))=1$, others must be sorted out. During the hashing process the server's guess can be calculated as:  
\begin{equation}
    \mathcal{P}(\mathcal{A}|\mathcal{B}) = \frac{1}{|\mathcal{DS}|-\sum\limits_{i=1}^{\mathcal{|DS|}}|P_\upsilon(h(s_i))-1|}
    \label{guess}
\end{equation}
This can reach it's maximum when the predicate function has been calculated for $\forall s_i\in \mathcal{DS}$, from which the lemma follows.
\end{proof}

\begin{figure}[ht]
\centering
\includegraphics[width=0.34\textwidth]{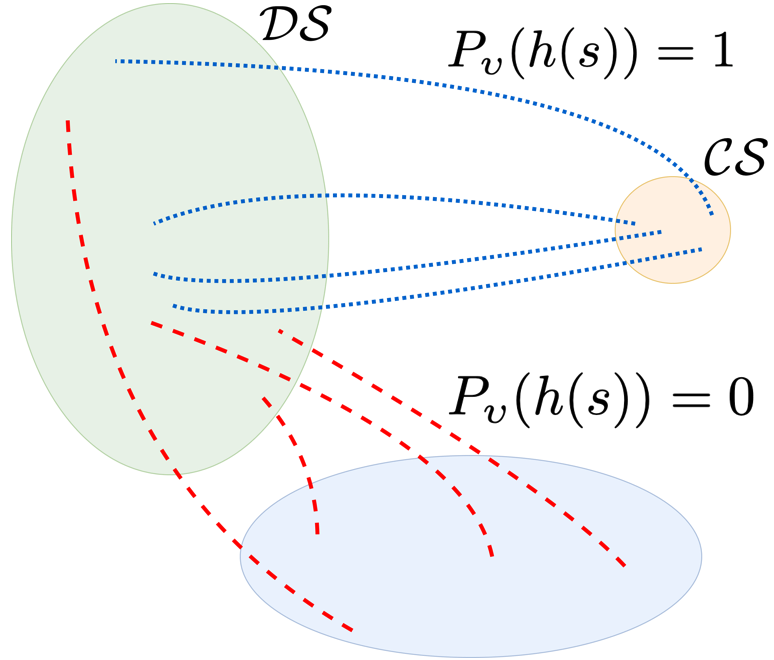}
\caption{Sorting the data set with the predicate function}
\label{sets}
\end{figure}
We mention again, that considering the probability distribution over the \textit{candidate set} should only be attempted after sorting has been performed.
From equation \ref{guess}, and Figure \ref{sets} it can be seen, that if for example the hashing process is almost complete, and only one element is left in $\mathcal{DS}$ that is not sorted, performing that last hashing operation can further decrease the size of the set we need to guess from.
Now it is obvious, that as the hashing process starts and after there are for example five candidates in $\mathcal{CS}$, the server cannot simply stop and say that now it has a $1/5$ chance to guess the cleartext. 

The next vital question to examine is what can the server know about condition $\mathcal{P}(\mathcal{B})$. In the traditional cracking scenario, it was always obvious when a certain \textit{data set} managed to crack a password, it was when a cleartext was found for the only hash digest. In our case there are always approximately $r$ cleartexts found.
One way to be certain that $\mathcal{P}(\mathcal{B})\approx 1$, is to increase the size of the \textit{data set} to $|\Sigma^l|$ unique elements. Then most passwords in $\mathcal{X}_\upsilon$ will be cracked regardless of how $N_\upsilon$ was chosen. As a consequence $\mathcal{P}(\mathcal{A})\approx1/\mathcal{CS}\approx1/N_\upsilon$.

Since it is completely impossible to hash $|\Sigma^l|$ elements with real hash functions, the server can never narrow down it's real guess to the \textit{candidate set} without making assumptions or learning implicit information. 
For example if the server thinks that the \textit{target hash} hides a human generated password, and the \textit{data set} contains a dictionary of popular words and mangling rules, it might assign a higher probability for condition $\mathcal{B}$. If the password however, is generated by a password manager then it is highly unlikely that such a \textit{data set} will recover the \textit{target hash}. 
If the server aims to increase $\mathcal{P(B)}$, even if $|\Sigma^l|$ cannot be reached, a mathematically viable approach is to increase the size of the \textit{cracking data set}. Even though this works perfectly in theory, it has little practical usefulness for the server:
After step SND-P, the number of decoy hashes $|\mathcal{X}_\upsilon|$ is fixed. In addition to the dictionary we requested ($\mathcal{DS}_1$), the server will perform further attacks with $d-1$ more \textit{data sets}.
\begin{equation}
\sum\limits_{i=1}^{d} |\mathcal{DS}_i| \frac{\mathcal{|X}_\upsilon|}{ |\Sigma^l|} \approx\sum\limits_{i=1}^{d} r_i
\label{growr}
\end{equation}
As expected for every \textit{data set}, the server will get $r$ \textit{candidate passwords}. Although this increases $\mathcal{P(B)}$, $\mathcal{P(A|B)}$ decreases as the \textit{candidate set} grows. Let
$k=|\mathcal{DS}_1 \cup \mathcal{DS}_2 \dots \cup \mathcal{DS}_d |$ where  $\mathcal{DS}_i \cap \mathcal{DS}_j  = \varnothing$ where $i \neq j$, $i,j \leq d \in \mathbb{N}^+  $  then,
\begin{equation} \label{eq11}
\lim_{\substack{k\to|\Sigma^l|}}\sum\limits_{i=1}^{d} r_i \to \mathcal{|X}_\upsilon|
\end{equation}

As a consequence of \ref{eq11}, without being able to make any assumptions, the server can not make a guess on the pre-image better than $1/|\mathcal{X}_\upsilon|$.

Using the same vector $\upsilon$ as in Toy Example 1, we can simulate how an attack looks like when further \textit{data sets} are used to crack elements in $\mathcal{X}_\upsilon$. The server now selects a different core dictionary containing $605834$ French passwords\footnote{https://github.com/clem9669/wordlists/blob/master/dictionnaire\_fr}. To this, the server adds mangling rules appending a digit, and a special character at the end of every password using the following command: \texttt{john --format=crc32 --wordlist=dictionnaire\_fr --mask=?w?d?s vectorv.txt}. Note, that in this case the vectorv.txt file contains all $5880$ hashes, as \textit{John the Ripper} doesn't support configuring  a predicate function, but with $|\mathcal{X}_\upsilon|$ being this small writing it's content to a file is not a concern. As we later see in the real life use case, writing all the \textit{decoy hashes} to a hard drive would be completely impossible. In total we expected $r=(5880\cdot605834\cdot10\cdot32)/16^8 \approx 265$ \textit{candidate passwords}.
Our cracking process returned $277$ candidates which is close to $r$. A few examples from the new candidate set can be seen in Table \ref{french}.
 
\begin{table}[H] 
\caption{Further cracking $\mathcal{X}_\upsilon$ from Example 1.}
\centering
\begin{tabular}{|c|c|c|c|} 
\hline
\textbf{\#} & \textbf{password} & \textbf{\#} & \textbf{password}\\
\hline
1&Abaigar9\^             & 190& ospedaletto3$<$                 \\ 
\hline
33& bréchaumont7$"$                        & 204& Proostdij6]        \\ 
\hline
46& Chassé1$+$ 3$+$          &207& québecoise4\*   \\ 
\hline
170& Montagne-Cherie7\}            &271& wadonville8(   \\ 
\hline
182& NorthCrawley9\&             & 277&éléphantiasis1[    \\ 
\hline
\end{tabular} 
\label{french}
\end{table}

Predicting, the password choosing behaviour of a single individual who we do not know anything about is not a straightforward task, especially if this behaviour can include using a password manager that generates random strings of unknown length. This leaves a malicious attacker with an almost impossible task when it is trying to determine $\mathcal{P(B)}$, where the only viable option is to try and learn further information to maximize $\mathcal{P(B)}$, and analyse the \textit{candidate set} in hopes to improve $\mathcal{P(A|B)}$.

\subsection{Side-channel information on the search space}
In the second toy example, if the server would somehow learn that we were after an 8 digit PIN code, it would be able to simply guess from the \textit{candidate set}. 
This is due to that in addition to $\mathcal{P(B)}=1$, an exhaustive search has been performed on $8$ digit PINs therefore, Lemma \ref{guesslemma} has been satisfied. Thus guessing the password becomes $\mathcal{P(A)}=1/\mathcal{|CS|}$. 

A possible mitigation strategy that works even if the server learns this information on $\mathcal{P(B)}$, is to increase the number of decoy hashes in $\mathcal{X}_\upsilon$ substantially. Thus, $|\mathcal{CS}|$ will grow as well, and even if the server guesses from the \textit{candidate list} it will contain too many entries. At this point we would like to emphasize again, that the 3PC protocol does not communicate such information to the server. 
If the client is concerned that the server could learn useful information by observing the \textit{data set}, the client can also change the \textit{key-space}, such that it defines a bigger search space.
Such a strategy for Toy Example 2 could be that the client asks all combinations of lower-case letters and digits ut to $8$, where $|\mathcal{DS}|=46^8$, or ask for an up to 9 character brute force with digits instead of exactly 8. Since this also raises the number of hashing operations, the decision on whether such a strategy is necessary, or feasible, depends on the security and privacy needs of the client. 

An important differentiator between the client and the server side, is that the client in some cases can possess a lot of implicit information. The client must take extra care not to leak this info through a badly constructed \textit{data set}, as it can impact the server's guess on the key space. If the password is hashed by \textit{bcrypt}, only a small \textit{data set} is selected that contains the core dictionary of words specific to the target. This can be name, street address, places lived, pet names, hobbies, favourite sports teams, etc., which is supplied with a set of mangling rules. Such information on its own can carry significant privacy risks both to the \textit{data subject}, who's \textit{PII} is leaked, and to the client side who would reveal implicit information to the server.
For the above reasons, revealing privacy sensitive information through a badly constructed data set, that would explicitly identify the entity who the \textit{target hash} belongs to, must be avoided. 
In contrast to the PIN code cracking example, where a increased \textit{data set} is beneficial, with the second example it is more than necessary to increase the \textit{data set}. 
For such data to remain unlinkable in this context, the \textit{client} must consider relevant \textit{privacy protection principles}, to ensure that the \textit{data set} serves as a sufficiently sized anonymity set of similar data \cite{sweeney_k-anonymity_2002,bayardo_data_2005,pfitzmann_anonymity_2008}.

\subsection{Server Side Guessing Based On Password Ranking}

In connection with Toy Example 2 we would like to examine if the server can achieve a $\mathcal{P(A|B)}>1/|\mathcal{CS}|$ guess by observing the candidate set. As we employed a brute-force attack, where most of the \textit{data set} consists of random strings this will provide \textit{perfect k-anonymity} if the \textit{target hash} also hides a random generated string. What if the \textit{target hash} hides a human generated password, would this \textit{candidate set} still offer sufficient privacy protection?
The security of human chosen PINs is a well researched topic \cite{markert_this_2020}, where 4 and 6 digit PINs were shown to follow a Zipf distribution \cite{wang_understanding_2017}.
Human chosen PINs are likely to follow non-random patterns such as ``12345678'', ``11111111'', a date structure like ``19930810'' or leet talk \cite{li_leet_2021}, such as ``13371337'' which reads ``leetleet''. 
If $\mathcal{X}_\upsilon$ is large these categories of \textit{pre-images} will also be represented in $\mathcal{CS}$, as the predicate function selects a random subset of $\mathcal{DS}$. However, this will not provide perfect k-anonymity as the majority of passwords in the candidate set will be a random looking string.
If an attacker tries to rank these PIN codes in the \textit{candidate list}, the \textit{target hash} can be in one of these human categories, or among the random strings where it is up to the assumptions of the server to try and determine which group is the one to pick from. 
In this example if the client wants to ensure k-anonymity even for human generated PINs, the best approach is to start with a \textit{cracking data set}, containing number codes only adhering to that format. As this is a smaller subset of all 8 digit number codes, $N_\upsilon$ can be chosen to be larger in order to provide a sufficiently sized anonymity set, where all \textit{candidate passwords} look human generated. Thus if the \textit{target hash} hides such a password it is now hidden in a set of similar data.

Referring back to Toy Example 1 where we employed a dictionary of human generated passwords as our \textit{cracking data set}, one could also argue that the returned candidates are not equally likely, as password choosing habits of individuals usually does not follow a uniform distribution over all \textit{cracking data sets} \cite{wang_birthday_2019}. 
If the \textit{data set} $\mathcal{DS}$ is a set of random equally likely strings, this would be true however, this cannot be generalised for all \textit{cracking data sets}. As discussed in connection with Toy Example 1, the \textit{Rock-You} database was shown by Wang et al. to follow a Zipf-like distribution \cite{wang_zipfs_2017}. As a consequence, employing such a cracking dictionary can result in some passwords in the \textit{candidate set} appearing more probable from the server's point of view.

As the \textit{candidate set} can only contain passwords from the given $\mathcal{DS}$ that was used for the cracking process, the ranking of passwords in the original dictionary (which the server possesses), can be applied to the \textit{candidate set} to establish an order. 
Since our $\mathcal{DS}$ in 3PC will only contain every password once, frequency ranking for example is only possible if knowledge about other frequency ranked databases are taken into account \footnote{https://github.com/danielmiessler/SecLists/blob/master/Passwords/Leaked-Databases/rockyou-withcount.txt.tar.gz}. 
Note, that Wang et al. used frequency ranking as a stepping stone from which the distribution was concluded \cite{wang_zipfs_2017}.
As we discussed in Section \ref{literature}, there are many possible strategies to rank passwords, where each can produce a different result based on previous knowledge, training data, and assumptions it is built upon. Without implicit knowledge on the use case, ranking could completely impair the server's guessing ability, where even choosing at random could yield better results. 
\begin{figure*}[ht] \label{threesets}
\centering
\includegraphics[width=0.85\textwidth]{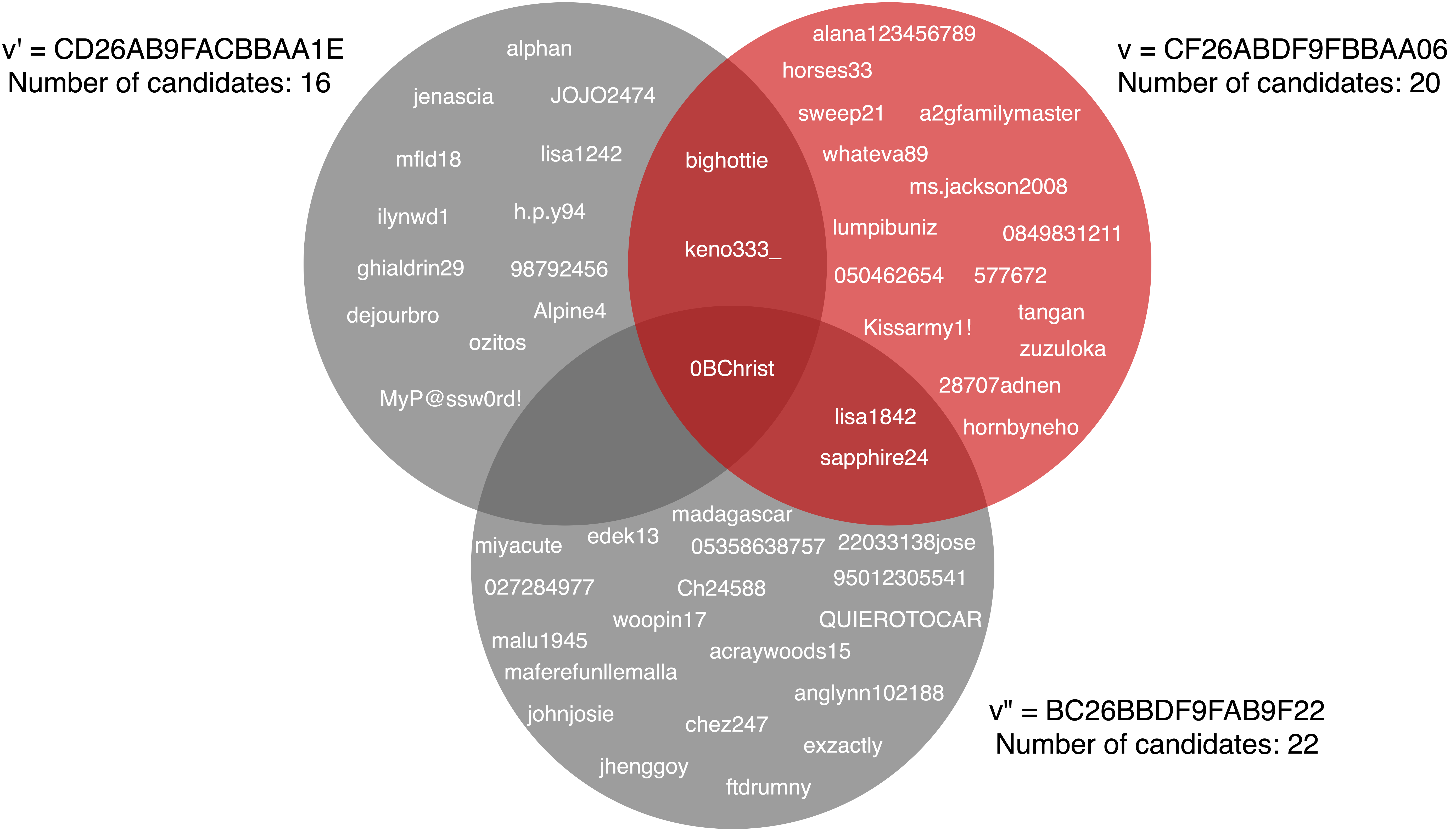}
\caption{An example on how a different vector $\upsilon$ defines a different random subset of $\mathcal{DS}$}
\label{examplev}
\end{figure*}

Ranking the passwords in $\mathcal{DS}$, and trying to guess the \textit{target hash} based on the likelihood of the returned cleartext passwords in $\mathcal{CS}$, will not improve the guessing ability of the server.
As we assume that the output of a hash function is uniformly distributed, the decoy passwords in $\mathcal{CS}$ are equally likely to be selected from $\mathcal{DS}$, regardless of what rank or likelihood is assigned to them. As such, the \textit{candidate passwords} are essentially a randomized subset of the cracking set $\mathcal{DS}$. Which randomized subset is picked, depends on how we selected our $\upsilon$. To put it differently, the vector $\upsilon$ will pre-determine which hashed \textit{pre-images} can satisfy the predicate function from $\mathcal{DS}$. However, this can never be known in advance without performing the hashing operation for every cleartext in $\mathcal{DS}$, as no correlation is assumed between cleartext passwords and their hashes. Thus, the occurrence of passwords from $\mathcal{DS}$ in the \textit{candidate set}, is not based on how likely those passwords were, or what rank they had. If they all appear in the \textit{data set} $\mathcal{DS}$ once, they are equally likely to show up in the \textit{candidate set} $\mathcal{CS}$. 
To demonstrate this in practice, we used the \textit{target hash} \texttt{0BChrist : C6BFABA2} from Toy Example 1, and created two additional vectors ($\upsilon'$ and $\upsilon"$), which contain the same \textit{target hash}.
Each vector $\upsilon$ defines a different $\mathcal{X}_\upsilon$ decoy set. These $\upsilon$ vectors over the \textit{Rock-You} data set will produce three different \textit{candidate sets} depicted in Figure \ref{examplev}. 

\begin{table}[h]
\caption{Frequency ranking}
\center
\begin{tabular}{|ccc|}
\hline
\multicolumn{3}{|c|}{\textbf{Frequency table}}                            \\ \hline
\multicolumn{1}{|c|}{$\boldsymbol{\upsilon}$} & \multicolumn{1}{c|}{$\boldsymbol{\upsilon'}$} & $\boldsymbol{\upsilon"}$  \\ \hline
\multicolumn{1}{|c|}{tangan : 14} & \multicolumn{1}{c|}{alphan : 4} & madagascar : 312  \\ \hline
\multicolumn{1}{|c|}{horses33 : 11} & \multicolumn{1}{c|}{jenascia : 2 } & sapphire24 : 4 \\ \hline
\multicolumn{1}{|c|}{sapphire24 : 4} & \multicolumn{1}{c|}{ozitos : 1} & miyacute : 2 \\ \hline

\end{tabular}
\label{toyfreq}

\end{table}

\noindent We employed a frequency analysis and a NIST password complexity \cite{choong_united_2014} check on all three \textit{candidate sets}. The results of the frequency analysis based on the original \textit{Rock-You} database can be seen in Table \ref{toyfreq}. As for the NIST guideline, there are only two passwords fulfilling the complexity requirements, which are \texttt{"Kissarmy1!"} and \texttt{"MyP@ssw0rd!"}, whereas $\upsilon"$ didn't contain a suitable candidate. Both of these passwords had a rank one frequency, so for each vector, these two ranking strategies produced completely contradictory results. Furthermore, the ranking strategies failed to pinpoint the \textit{target hash} as a likely candidate, for all three vectors. 

Here we would like to underline a crucial idea: the client must never generate more than one $\upsilon$ vector using GEN-V, as the intersection of the corresponding $\mathcal{X}_\upsilon$ sets can narrow the search for the \textit{target hash}. Our goal here was to demonstrate that different $\upsilon$ vectors produce different candidate sets over a specific $\mathcal{DS}$.
As in general, it is not possible to accurately predict the behaviour of hypothetical individuals, who we do not know anything about, the server cannot make assumptions on what ranking strategy to use. Hence, in this case we rely on empirical evidence and highlight that applying different ranking strategies could easily mislead the server, regardless of which random subset of $\mathcal{DS}$ was selected by $\upsilon$.

\subsection{Implicit information disclosure}

Imagine a scenario, where the \textit{NATO Communications and Information Agency} (\textit{NCI Agency}) wanted to recover the cleartext for the \texttt{"58727AD23361CA0323D4B3C22A6AFE78"} \textit{NTLM hash} using a specialized \textit{PCaaS} company. After the cracking process the server side observes the \textit{candidate set}, and finds the cleartext password \texttt{Bices2014}, among the other password candidates. 
The \textit{Battlefield, Information, Collection and Exploitation Systems} (\textit{BICES}) serves as the primary intelligence-sharing network between and among all 28 \textit{NATO} member nations, seven associated partner nations and the \textit{NATO organization}.
Upon learning that the customer is NCI Agency, a potential attacker could extrapolate the following information:
\begin{itemize}
\item This password could belong to a classified \textit{NATO system} called \textit{BICES};
\item The use of special characters is not enforced by the password policy;
\item The password was possibly not changed in the last 8 years, in which case no password change policy is enforced.
\end{itemize}
Learning the client's identity is a form of information disclosure that can serve as a starting point for attackers.
\textit{NCI Agency} follows best practices that the hash of a password is confidential and should not leave the security boundary of the organization \cite{ewaida_pass--hash_2010}.
As a standalone password without the direct connection to NATO, \texttt{Bices2014} does not convey any useful information to the server. The word Bices can be associated with various concepts. For example in the Haitian \textit{Creole} language Bices means bicycles and is spelled exactly the same. Similarly, Bices is the abbreviation for a mining seminar organized annually in China.
By design, the 3PC protocol does not allow such information to be exchanged between the server and the client. However, the client must take precaution not to leak this via other channels. This is an example for information disclosure outside of the 3PC protocol. 

Although this looks like a high risk scenario on the first glance, everything depends on how the \textit{data set} was chosen.
Let the selected core dictionary contain common words, and ones related to NCI Agency, such as; NATO, BICES, North, Atlantic, Treaty, Organization, NCI, Agency, Jens, Soltenberg, etc. In addition, this can be complemented with a set of mangling rules such as; capitalizing the first letter of every key-word, appending one to three special characters, and optionally replacing relevant letters in the core words to form leet talk like N4t0. A \textit{data set} constructed with such rules, will not only have a lot of relevant passwords, but it can be argued that the probability distribution is close to uniform, as these passwords look equally likely from the server's perspective. The security parameter $r$ can be adjusted such that the \textit{candidate set} can contain millions of passwords in this format, like NC1N4T0.2022, Bic3\$pass123 etc.

What we deem acceptable from a privacy perspective greatly depends on the use case. If a company tests corporate passwords of employees in important roles, getting the password cracked by common dictionaries is a good scenario. To discover and change a weak password, is better then the inevitable data breach. 
Note, that the password is only cracked by such a \textit{cracking data set}, if the cleartext happens to be in this format.
The more difficult question comes when the password is not cracked. What if the third party continues cracking with different dictionaries, brute force attacks etc. In this case, the server runs into the problem of growing \textit{candidate sets} as we previously discussed in relation to equation \ref{growr}.

\subsection{Foul Play}
When the client side starts step CHK-CS of the protocol to see if a \textit{pre-image} for the target hash was found, as a part of this step it must investigate whether the server indeed exhausted the agreed search space and performed $|\mathcal{DS}|$ hashing operations. To confirm this, the client relies on \textit{proof of work}, which is a concept where a \textit{prover} demonstrates to a \textit{verifier} that a certain amount of a computational effort has been expended in a specified interval of time \cite{lachtar_cross-stack_2020}. In the case of 3PC the client knows that approximately $r$ candidates must be returned based on the formula, $|\mathcal{X}_\upsilon| \frac{|\mathcal{DS}|}{|\Sigma^l|} \approx r$, where $P_\upsilon(x)=1$ must be satisfied for all hash digests.

The server cannot simply fill the candidate set with randomly chosen $r$ hash digests as they need to fulfill the \textit{predicate function} for the given vector $\upsilon$. If the server selects fitting hash digests with fake \textit{cleartext passwords}, the client can simply select a random subset in $\mathcal{CS}$, and hash it to verify if it is indeed the correct one. Obviously, $r$ is only an expected value but deviating from it significantly can suggest that the server is not truthful in expending the appropriate resources.

\subsection{Plausible deniability}
If any entity upon acquiring the vector $\upsilon$ or the \textit{candidate set}, would claim that the client was trying to break a password hash $t'$ that belongs to them, they can not incriminate the client. Even if it is true that $t'\in\mathcal{X}_\upsilon$, the client can always rely on plausible deniability and say that this happened by chance. Indeed, when a specific vector $\upsilon$ is calculated for a given \textit{target hash}, it is true that any $x\in \mathcal{X}_\upsilon$ could have been the seed given to GEN-V, that produced vector $\upsilon$.
This means that the probability that an arbitrary hash digest from $\Sigma^l$ falls into $\mathcal{X}_\upsilon$ is $\frac{|\mathcal{X}_\upsilon|}{|\Sigma^l|}$, which is a non-negligible probability if the parameters were chosen properly. The client, before transferring $\upsilon$, can check if the chosen $N_\upsilon$ provides  satisfactory plausible deniability based on the scenario.

\subsection{Zero-Knowledge variation}
If the \textit{cracking data set} is small, it enables the 3PC protocol to be used in a \textit{zero-knowledge} setting.
As we noted, the upper limit for the size of $\mathcal{X}_\upsilon$, is how many \textit{candidate passwords} can be transferred and stored from a given \textit{data set}. 
If the client only needs a small \textit{cracking data set}, where it would be possible to store and transfer $|\mathcal{DS}|$ hashes and the corresponding passwords, then $N_\upsilon$ can be selected such that:  $N_\upsilon=|\mathcal{X}_\upsilon|=|\Sigma^l|$ i.e., $\{\upsilon_{2i}=F \, \wedge \, \upsilon_{2i-1}=0: \forall \,1 \leq i \leq l \}$.
Now, there is no need to calculate the predicate function as $P_\upsilon(h(s))=1, \forall s\in \Theta^*$ as the decoy set is essentially the output space of $h$. It is then trivial, that not sending the \textit{target hash} at all, or sending it in vector $\upsilon$ with the above described \textit{zero-knowledge} setting, is equivalent.

\section{Real life experiments} \label{experimantal}

In this experiment we are looking to empirically verify the following question: Is it really feasible to crack billions of hashes using the theory introduced in the 3PC protocol? In the following section, we show the protocol running on an FPGA architecture, where the second part of this section will aim to analyse and interpret the results and related privacy implications. 

\subsection{Implementation on the RIVYERA FPGA cluster}
The Cost-Optimized Parallel Code Breaker (COPACOBANA) FPGA hardware was introduced in the annual Conference on Cryptographic Hardware and Embedded Systems (CHES) in 2006 \cite{kumar_breaking_2006}. The RIVYERA FPGA hardware architecture \cite{rivyera_rivyera_s6_lx150_rev431_datasheet_nodate} is the direct successor of COPACOBANA, equipped with Spartan–6 Family FPGA modules which can be seen in Figure \ref{RIVYERA}. The demonstration of the 3PC protocol is conducted on a RIVYERA S6-LX150 cluster, containing $256$ Xilinx Spartan–6 LX150 FPGA modules. 

\begin{figure}[ht]
\centering
\includegraphics[width=0.35\textwidth]{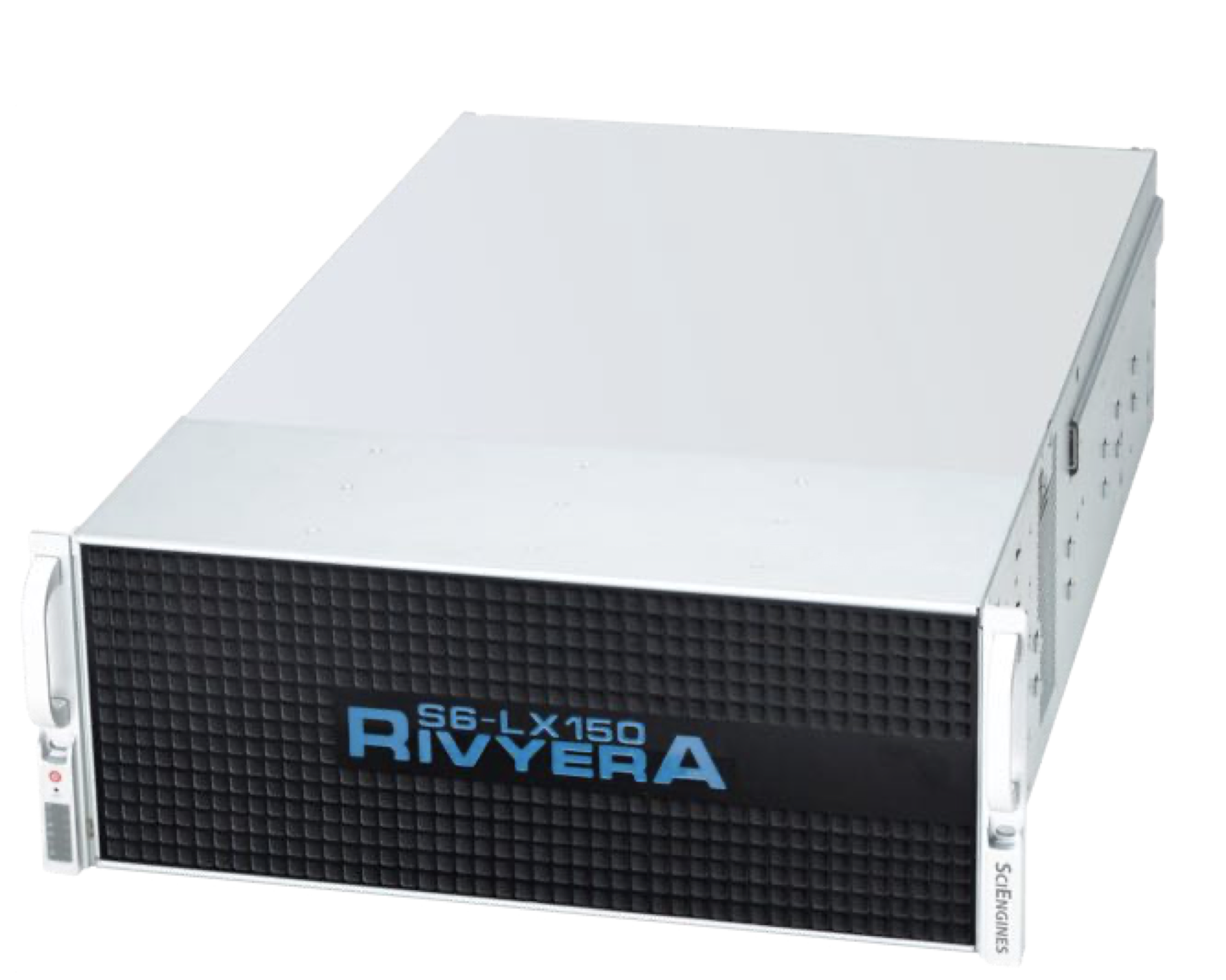}
\caption{RIVYERA S6-LX150 server \cite{rivyera_rivyera_s6_lx150_rev431_datasheet_nodate}}
\label{RIVYERA}
\end{figure}

RIVYERA is using the \textit{se\_decrypt} 3.00.08 cryptanalysis framework, developed by SciEngine\footnote{https://www.sciengines.com/it-security-solutions/cryptanalysis-tools/}, which is a password cracking tool designed to maximize the efficiency of the FPGA modules. This tool allows configuring the predicate function for a given  vector $\upsilon$, without any design changes or modification in the software.
As discussed in the introduction, the main motivation behind this paper originates from a penetration testing engagement where the Red Team was able to retrieve an \textit{NTLM} hash for an important service account. 
It was known that all service account passwords are randomly generated 9 character long strings containing uppercase and lowercase letters plus numbers. According to the signed contract revealing any cleartext passwords to third parties was prohibited. 
However, the client made an exception and approved the use of a PCaaS, if the Red-Team can assure that the third party server does not have a guess with better than $2^{-29}$ probability. 
The following example presents a realistic scenario that solves this problem using 3PC.

Let the data set $\mathcal{DS}$ be the set of all 9 character long strings containing uppercase and lowercase letters plus numbers hence, $|\mathcal{DS}|=62^9$. 
According to the specified security requirements, the server must not have a better guess than $1/\mathcal{|CS|} < 2^{-29}$. As a reminder, the 3PC protocol provides a stronger security, as the server can not make a $1/|\mathcal{CS}|$ guess on the \textit{candidate set} without making several assumptions. 
The client wanted assurance that even in the worst case scenario if $s^* \in \mathcal{DS}$ is known (in other words $\mathcal{P(B)}=1$), the maximum guessing probability of the server is $1/|\mathcal{CS}|$. 
In this example the NTLM \textit{target hash} is $t= $\texttt{8AC54208A85C340AE9B8B0CDB236F14C}. 

Compared to the GEN-V step in the 3PC protocol, where it was possible to set a degree of freedom on each hexadecimal character of the target hash, the \texttt{se\_decrypt} tool is more limited. The built-in \texttt{--hit-mask} parameter can be used as a "restricted" GEN-V function to create $\upsilon$. It only allows setting the degree of freedom by bytes (by two hex characters). 
To be more precise, the n-th bit of the hit mask (from the right) corresponds to the n-th byte of the target (also from the right). If the n-th bit of the hit mask is a "1", the n-th byte of a resulting \textit{candidate password hash} must completely match the n-th byte of the \textit{target hash}, otherwise it does not need to match. The \texttt{se\_decrypt} expects this vector $\upsilon$, and the \texttt{--hit-mask} both in hexadecimal representation. As a consequence, we can only select even powers of sixteen as the size of $N_\upsilon$.
To transform the vector $\upsilon$ which is suitable for the \texttt{se\_decrypt} tool, one needs to solve the following RIVYERA specific inequality:

\begin{equation}
\label{RIVYERAeq}
    2^{29} < \frac{|\mathcal{DS}| \cdot 16^{2x} }{|\Sigma^l|}
\end{equation}
The first $x$ which satisfies \ref{RIVYERAeq} is $x=13$, for which we have 

\begin{equation*}
   \frac{62^9 \cdot 16^{26}}{16^{32}} \approx 806\,873\,234  
\end{equation*}
Thus we expect  approximately $r \approx 806$ million candidates after hashing all possible strings in $\mathcal{DS}$. A suitable vector is $\upsilon=$\texttt{\seqsplit{[88AACC550F0F0F0F0F0F0F0F0F0F0F0F0F0F0F0F0F0F0F0F0F0F0F0F0F0F44CC]}} $ \in \Sigma^{64}$.
One can easily generate the RIVYERA specific mask on the server side namely \texttt{8AC5000000000000000000000000004C}. 
This shows that from the $32$ hexadecimal character long target hash we need to hide 26 characters ($16^{26}$) in any position. Using the \textit{se\_decrypt} tool, one can use the following parameters to start cracking the data set: \texttt{se\_decrypt -a nthash --hash 8AC5000000000000000000000000004C --hit-mask C001 --lower-case --upper-case --numbers --min-len 9 --max-len 9 --no-stop
 ---logfile candidate.txt} 
 
\noindent
Please note, that the binary representation of the \texttt{C001} parameter value is \texttt{1100000000000001}, meaning we disregard 13 bytes from the middle (26 hexadecimal characters) from the original hash. The \texttt{8AC5000000000000000000000000004C} hash together with the \texttt{--hit-mask C001} produces the same decoy hashes as the original vector $\upsilon$, therefore these are equivalent representations of the $\mathcal{X}_\upsilon$ \textit{decoy set}.

In our control test, brute forcing the whole $62^9$ key space for one NTLM hash takes approximately 19.6 hours on RIVYERA with $192$ billion hash/sec. Here comes the real benefit of the 3PC protocol: After applying the \texttt{--hit-mask C001} parameter on our modified vector $\upsilon$ defining $|X_\upsilon|=16^{26}$ hashes, the RIVYERA cluster finished the cracking process after $19.6$ hours, under the same time as it took to find a single hash, and found $806834341$ \textit{candidate passwords}. Writing all candidates to disk took less than 50 seconds which is negligible compared to the cracking time. The output of the cracking process can be seen in Figure \ref{Terminal}. 
The last step is to send the resulting file to the client, which is around $20.0$ GB. The client then simply checks whether $\exists(s,t) \in \mathcal{CS}$, which is true in our case. The cleartext corresponding to the \textit{target hash} $t$ is \texttt{"bKFQ4Q8C0"}. 

\begin{figure}[ht]
\label{RIVYERAcracked}
\centering
\includegraphics[width=0.4\textwidth]{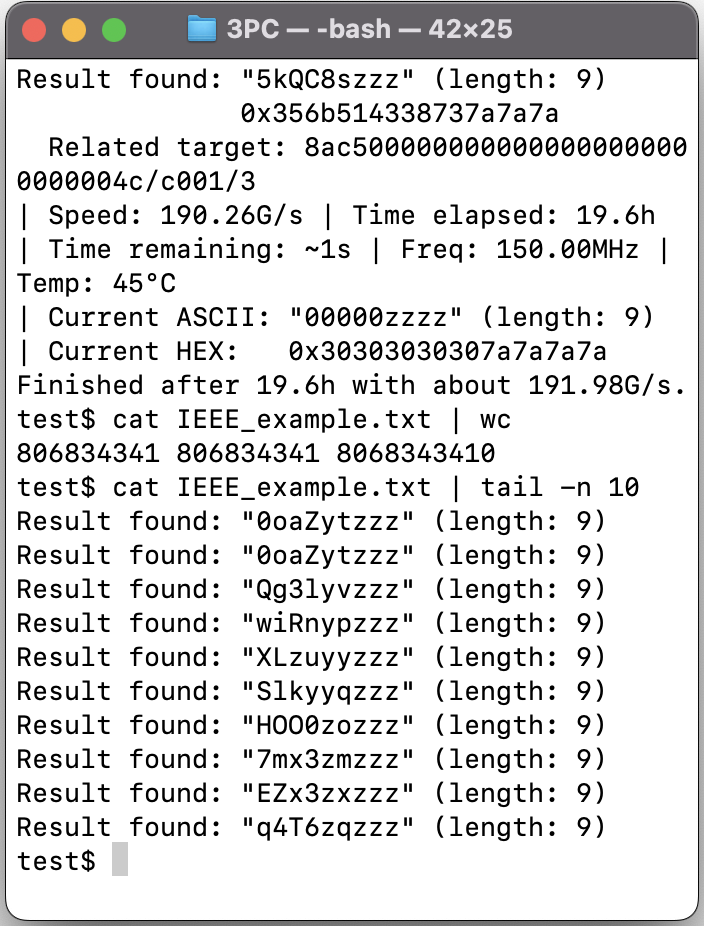}
\caption{Bash terminal output on the cracking process}
\label{Terminal}
\end{figure}

\subsection{Experimental result analysis}
A very clear advantage of 3PC is that using a predicate function allows the efficient handling of huge anonymity sets, that would otherwise be infeasible. In our experiment presented on RIVYERA, the size of the \textit{decoy set} was selected to be $|\mathcal{X}_\upsilon| = 16^{26}=20282409603651670423947251286016$. Applying the $P_\upsilon$ function to the hashed \textit{data set} produced an anonymity set of $r \approx 806$ million decoy passwords. In this example $\mathcal{CS}$ provides \textit{perfect k-anonimity} on the \textit{target hash}, as the password is hidden in a set of similar data. 

Despite the \textit{decoy set} being large, the 3PC protocol utilizes the predicate function to allow a constant time check for verifying $h(s)\in \mathcal{X}_\upsilon,$ (based on lemma \ref{ordo1}) for $\forall s\in\mathcal{DS}$.
As a result we can increase $|\mathcal{X}_\upsilon|$ without any penalties on computational speed. The only upper limit is the size of the \textit{candidate set} which the server needs to write to disk. 
In NTLM, one hash value is exactly 128 bit. From this we can easily calculate how much hard drive space it would take if one would attempt to write the content of the \textit{decoy set} to the hard drive: $(16^{26} \cdot 128) /8/1024^5 \approx 2.88 \cdot 10^{17}$, which is 3 million billion Petabytes. This would clearly be impossible to handle with today's technology. However 800 million cleartext-hash pairs in the \textit{candidate set} only takes up $20$ GB, which can be easily transferred to the client. 
To the best of our knowledge this example was the first documented case when perfect \textit{k-anonimity} was achieved during a password cracking process using a third party service.

\subsection{Password analysis and Parameter security}
What if the client was not sure if the \textit{target hash} hides a machine generated code? Would it be safe to use a brute force attack, assuming that the \textit{candidate set} provides a sufficient anonymity set for a password like \texttt{JohnnY007}?
To answer this question we have analyzed the 800 million \textit{candidate passwords} which are a random subset of the \textit{data set}, selected by the $P_\upsilon$ function. We have used the \texttt{cracklib-check} plugin, which groups passwords into six categories. From the 800 million candidates approximately $0.07\%$ of the passwords were flagged as weak.
The exact distribution for each category can be seen in TABLE \ref{passdistr}.
\begin{table}[ht]

\caption{Strength distribution of randomly generated passwords}
\centering
\begin{tabular}{|l|l|l|}
\hline
\textbf{Category}  & \textbf{Number of passwords} & \textbf{Percentage} \\ \hline
The password is random (\texttt{OK})       & 806263493           & 99.93\%   \\ \hline
The password is too simple   & 532747              & 0.066\%     \\ \hline
Based on a dictionary word     & 11763               & 0.00146\%   \\ \hline
Reversed dictionary word  & 11895               & 0.00147\%   \\ \hline
Not enough different characters  & 6111                & 0.00076    \\ \hline
Looks like an Insurance number      & 8332                & 0.001\%    \\ \hline
\end{tabular}
\label{passdistr}
\end{table}

As an example we list below some of the dictionary based words that the tool found.

\begin{itemize}
    \item \texttt{WnbBATman} : This password can be interpreted as the famous superhero of the Warner Brothers. "WnB Batman".
    \item \texttt{MRloNDon6}: Starts with a Title, "MR", followed by a dictionary word, and ends with a number. "Mr London 6" can also be mistaken as human generated. 
    \item \texttt{PaYpaLQSC}: What about a password used by employees from the Paypal Quality Service Center?
    \item \texttt{goOgLEFOG}: This can be interpreted as the  Google Cloud Platform API solution "Google Fog".
\end{itemize}
Although there are a good number of passwords that can be easily classified as human generated, this set would not provide perfect k-anonymity for all human generated passwords.
The occurrence of shorter dictionary words is higher, but passwords that use all 9 characters in a plausible sequence are less common. This is to be expected, as the protocol is designed to provide an anonymity set for data that makes up the \textit{cracking data set}.
Therefore, we provide the following guidelines for practical applications.

If the client possesses no implicit information on the \textit{pre-image} of the \textit{target hash}, the best approach is to start with a hybrid attack, based on a core dictionary and a rule-set. A smaller \textit{data set} with a larger $N_\upsilon$, for which the client can repeatedly switch between the \textit{data sets} of similar size, without changing $\mathcal{X}_\upsilon$ has an added benefit.
This will make it impossible for the server to conduct attacks with significantly bigger \textit{data sets} such as brute force attacks, as the resulting \textit{candidate set} for the same $N_\upsilon$ could reach sizes of hundreds of Petabytes, making it impossible to store or evaluate the results.
If these cracking sessions are proven to be unfruitful, the client could choose to pivot to brute force attacks. If this entails a significantly larger \textit{data set}, the client can shrink $\mathcal{X}_\upsilon$, but never create a new vector $\upsilon$ through GEN-V.

\section{Conclusion} \label{summary}
This paper introduces privacy-preserving password cracking, showing how the computational resources of an untrusted third party can be used to crack a password hash. The anonymity sets that hide the target hash, would make it impossible in traditional cracking scenarios to process or store this number of hashes and the corresponding cleartext results. We circumvented this by extending the theory of predicate functions to operate on the output of hash functions.
On top of this, we demonstrated that increasing the number of decoy hashes bears no impact on the hash rate, making the 3PC protocol very efficient. 
The implementation of the protocol was shown through two Toy Examples, and one real life implementation running on the RIVYERA FGPA cluster.
Through these, we were able to verify our original goals we set out to examine.
The server only learns probabilistic information both on the \textit{target hash} and the \textit{cleartext password}. The protocol is resistant against \textit{foul play}, where the server gains no tangible advantage towards learning the target hash, or it's \textit{pre-image} by not following the steps of the protocol. The protocol ensures \textit{plausible deniability}, where the client can claim to have aimed for a different target. Through a proof of work scheme, the client can have a statistical argument if they suspect that the requested search space has not been exhausted.
The empirical tests and the theoretical analysis suggests that the 3PC protocol is suitable for practical use both from a security, a privacy, and an efficiency perspective.

\appendices

\ifCLASSOPTIONcompsoc
  \section*{Acknowledgments}
\else
  \section*{Acknowledgment}
\fi

The authors would like to thank Dr. Axel Y. Poschmann for the constructive brainstorming sessions, and Dr. Lothar Fritsch for his insightful comments.
  
\ifCLASSOPTIONcaptionsoff
  \newpage
\fi


\end{document}